\documentclass[11pt]{article}
\usepackage[margin=1in]{geometry}
\usepackage{appendix}
\usepackage{babel}
\usepackage{amsthm}
\usepackage{graphicx}
\usepackage{amsmath}
\usepackage{amssymb}
\usepackage{array}
\usepackage[colorlinks=true, citecolor=black, linkcolor=black]{hyperref}
\usepackage{authblk}

\usepackage[linesnumbered,ruled,vlined]{algorithm2e}
\usepackage{algorithmicx}
\usepackage{bm}
\usepackage{tikz}
\usepackage{subcaption}
\usepackage{tabulary}
\usepackage{longtable}
\usepackage{url}
\usepackage{cleveref}
\usepackage{natbib}
\let\cite\citep
\let\shortcite\citeyearpar

\newtheorem{theorem}{Theorem}[section]

\newtheorem{lemma}[theorem]{Lemma}

\newtheorem{definition}{Definition}[section]
\theoremstyle{remark}
\newtheorem{remark}{Remark}[section]
\newtheorem{example}{Example}[section]

\newcommand{\nobracketcite}[1]{\citeauthor{#1} \citeyear{#1}}
\newcommand{\noindentparagraph}{\paragraph}
\newcommand\class[1]{$\mathsf{#1}$}

\newcommand{\domain}{\Delta}
\newcommand{\sentence}{\Psi}
\newcommand{\uevidence}{\mathcal{U}}
\newcommand{\binevidence}{\mathcal{E}}
\newcommand{\otc}{\bm{\zeta}} % one-type configuration
\newcommand{\otce}{\zeta} % one-type configuration element
\newcommand{\pred}[1]{\mathsf{pred}\left( #1 \right)}
\newcommand{\fomodels}[2]{\mathcal{M}_{#1, #2}}

\newcommand{\weight}{w}
\newcommand{\negweight}{\overline{w}}
\newcommand\wfomc{\mathsf{WFOMC}}

\newcommand\refformula[1]{\ensuremath{{#1}_{\text{ref}}}}

\newcommand\FOtwo{$\text{FO}^2$}
\newcommand\UFOtwo{$\text{UFO}^2$}
\newcommand\FOthree{$\text{FO}^3$}
\newcommand\Ctwo{$\text{C}^2$}

\title{Tractable Weighted First-Order Model Counting \\ with Bounded Treewidth Binary Evidence}
\author[1]{V\'aclav K\r{u}la}\affil{Czech Technical University in Prague, Prague, Czech Republic}
\author[2]{Qipeng Kuang}\affil{The University of Hong Kong, Hong Kong, China}
\author[3]{Yuyi Wang}\affil{CRRC Zhuzhou Institute, Zhuzhou, China}
\author[4]{Yuanhong Wang}\affil{Jilin University, Changchun, China}
\author[1]{Ond\v{r}ej Ku\v{z}elka}
\date{}

\begin{document}

\maketitle

\begin{abstract}
The Weighted First-Order Model Counting Problem (WFOMC) asks to compute the weighted sum of models of a given first-order logic sentence over a given domain. Conditioning WFOMC on evidence---fixing the truth values of a set of ground literals---has been shown impossible in time polynomial in the domain size (unless \class{\#P \subseteq FP}) even for fragments of logic that are otherwise tractable for WFOMC without evidence. In this work, we address the barrier by restricting the binary evidence to the case where the underlying Gaifman graph has bounded treewidth. We present a polynomial-time algorithm in the domain size for computing WFOMC for the two-variable fragments \FOtwo{} and \Ctwo{} conditioned on such binary evidence. Furthermore, we show the applicability of our algorithm in combinatorial problems by solving the stable seating arrangement problem on bounded-treewidth graphs of bounded degree, which was an open problem. We also conducted experiments to show the scalability of our algorithm compared to the existing model counting solvers.
\end{abstract}

\section{Introduction}

The Weighted First-Order Model Counting Problem (WFOMC) asks to compute the weighted sum of models for a given first-order logic sentence over a specified domain, alongside a pair of weighting functions that assign a weight to each model of the sentence. WFOMC serves as a fundamental problem in Statistical Relational Learning \cite{getoor2007introduction} where applications such as Markov Logic Networks \cite{WFOMC-UFO2}, parfactor graphs \cite{poole2003first}, probabilistic logic programs \cite{WFOMC-FO2} and probabilistic databases \cite{VandenBroeck13} can be reduced to WFOMC. Recent work also reveals the potential of WFOMC to contribute to enumerative combinatorics by providing a general framework for encoding counting problems \cite{WFOMC-axioms-recursion}, integer sequences \cite{fluffy}, and computing graph polynomials on specific graphs \cite{WFOMC-polys}.

In WFOMC we usually measure the time complexity in terms of the domain size. The fragments enjoying polynomial time complexity in the domain size are called \emph{domain-liftable} \cite{WFOMC-UFO2}. Previous results have shown that the two-variable fragment (\FOtwo{}) possibly with counting quantifiers (\Ctwo{}) and cardinality constraints is domain-liftable \cite{WFOMC-UFO2,WFOMC-FO2,WFOMC-C2}.

We consider the task of computing WFOMC for \FOtwo{} and \Ctwo{} with \emph{evidence}, the interpretation of certain predicates. For example, consider the sentence
\begin{equation}\label{eq:indset}
\sentence = \forall x \forall y: E(x,y) \to (\lnot I(x) \lor \lnot I(y))
\end{equation}
over the domain $\domain = \{1,2,3,4\}$. An evidence can be the specification that $E(1,2), E(1,3), E(2,3), \\ E(1,4)$ and their reverse atoms are true and other atoms of $E$ are false. In this case, each model is an assignment of $I(1), \cdots, I(4)$ satisfying $\sentence$, which represents an independent set of the graph with vertices $\{1,2,3,4\}$ and undirected edges $(1,2), (1,3), (2,3), (1,4)$. The WFOMC of $\sentence$ with the above evidence over $\domain$ with unit weight for each model equals to the number of independent sets on the graph.
%Conditioning WFOMC on evidence enhances its expressive power and applicability.

It has been shown that \FOtwo{} and \Ctwo{} remain domain-liftable when conditioned on unary evidence (i.e., interpretation of unary predicates) \cite{WFOMC-conditioning,WFOMS-FO2-journal}. In the contrary, conditioning on binary evidence (i.e., interpretation of binary predicates) is \class{\#P}-hard, even for the universally quantified \FOtwo{} sentences \cite{WFOMC-conditioning}.
An intuition of this phenomenon is that the binary evidence breaks the \emph{symmetry} of domain elements. Without binary evidence, domain elements can be partitioned into a fixed number of classes, where elements within the same class are indistinguishable. For example, for the sentence in \cref{eq:indset} conjuncted with $\forall x: \lnot E(x,x)$, elements can be classified based on whether $I(x)$ holds. This symmetry is crucial for current algorithms computing WFOMC of domain-liftable fragments. Unary evidence does not break the symmetry of elements, as they can be eliminated using the technique proposed by \citet[Appendix A]{WFOMS-FO2-journal}, allowing WFOMC to be computed on a first-order sentence without unary evidence.

Little progress has been made to deal with the binary evidence. \citet{WFOMC-binary-evidence-bmf} showed that if the binary evidence can be represented by a binary matrix with low boolean rank, then a boolean matrix factorization \cite{bmf} transfers the binary evidence to unary evidence. However, as stated by the authors, real-world binary matrices are likely to have large boolean rank, thus the use of this approach is limited in reality.

\subsection{Our Contributions}

%In this work, we study the binary evidence whose Gaifman graph, the undirected graph with vertices corresponding to domain elements in which two vertices $a$ and $b$ are connected if and only if there is a binary predicate $P$ such that the truth value of $P(a,b)$ or $P(b,a)$ is specified in the binary evidence, is of bounded treewidth.
%Our primary result is to show that computing WFOMC for \FOtwo{} and \Ctwo{} with such binary evidence can be done in time polynomial in the domain size. Our approach applies to counting problems in combinatorics on bounded-treewidth graphs by encoding the problem by WFOMC of a fixed \FOtwo{} or \Ctwo{} sentence with such binary evidence.
We propose an algorithm to compute WFOMC for an \FOtwo{} or \Ctwo{} sentence with cardinality constraints, unary evidence and binary evidence whose Gaifman graph, the undirected graph with vertices corresponding to domain elements in which two vertices $a$ and $b$ are connected if and only if there is a binary predicate $P$ such that the truth value of $P(a,b)$ or $P(b,a)$ is specified in the binary evidence, is of bounded treewidth. The algorithm runs in time polynomial in the domain size. This offers a new approach to overcome the symmetry limitation in lifted inference, and expands the expressive power and applicability of WFOMC.
The key idea of our algorithm is the dynamic programming for computing WFOMC for the universally quantified \FOtwo{} sentence with evidence on the tree decomposition of the underlying Gaifman graph. Existing techniques help to extend the result to \FOtwo{} and \Ctwo{} with cardinality constraints.

Experiments are conducted to test the performance of our algorithm compared to two existing model counters on two problems, the friends and smokers problem and the inference on Watts-Strogatz graphs. Results indicate the efficiency and scalability of our algorithm.

Finally, our approach applies to combinatorial problems on bounded treewidth graphs by solving an open problem \cite{Roger_2023}: the stable seating arrangement problem of fixed number of classes on graphs of bounded treewidth and bounded degree. We show that counting the number of stable seating arrangements in such case is in time polynomial in the number of agents.

\begin{remark}
A similar result is Courcelle's theorem \cite{Courcelle90} which states that the model checking problem of monadic second order logic on a bounded treewidth graph is in polynomial time. We remark that our result and Courcelle's theorem are incomparable. First, as Courcelle's theorem works with {\em monadic second order logic of graphs}, it is required that the interpretation to all but unary predicates should be given. In contrast, we allow binary predicates to remain uninterpreted. Second, our result is restricted to the two-variable fragment, which is a necessary consequence of the negative result showing intractability of WFOMC for \FOthree{} \citep{WFOMC-FO3}, while Courcelle's theorem can involve arbitrary number of variables because it works in a different setting.
\end{remark}

\section{Preliminaries}

\subsection{First-Order Logic}
\label{subsec:fol}

We work with the function-free fragments of first-order logic (FOL).
%A language $\mathcal{L}$ consists of a finite set of constants $\Delta$, a finite set of variables $\mathcal{V}$, and a finite set of predicates $\mathcal{P}$.
%We denote a predicate $P$ of arity $k$ as $P/k$.
Let $\mathcal{P}$ be the vocabulary of \emph{predicates} (also called \emph{relations}).
An \emph{atom} has the form $P(t_1, t_2, \cdots, t_k)$ where $P \in \mathcal{P}$ is a predicate and $t_i$ is a logical \emph{variable} or a \emph{constant}.
A \emph{literal} is an atom or its negation.
A \emph{formula} is a literal, or created by connecting formulas using negation, conjunction or disjunction, or by quantifying a formula using a universal quantifier $\forall x$ or an existential quantifier $\exists x$ where $x$ is a logical variable.
A \emph{ground} formula is a formula containing no variables.
A variable in a formula is called \emph{free} if there is no quantification over that variable.
A \emph{sentence} is a formula with no free variables.

%We adopt the Herbrand semantics ***REF*** with a finite domain.
%We denote the domain with $\Delta$ as there is a one to one correspondence to the set of constants.
%We denote the Herbrand base by HB.
%We use $\omega$ to denote a possible world, i.e., any subset of HB.

%We work with the first-order logic fragment with at most two variables, called \FOtwo{}.
%This fragment allows for lifted inference.

% A first-order logic formula in which all variables are substituted by constants in the domain is called \emph{ground}.
A \emph{possible world} $\omega$ interprets each relation in a sentence over a finite domain, represented by a set of ground literals. We write $\omega \models \alpha$ to denote that the formula $\alpha$ is true in $\omega$ following the standard FOL semantics. The possible world $\omega$ is a \emph{model} of a sentence $\sentence$ if $\omega \models \sentence$. We denote the set of all models of a sentence $\sentence$ over the domain $\domain$ by $\fomodels{\sentence}{\domain}$.

In this paper, we are specially interested in the following syntactic fragments of FOL. A sentence with at most two logical variables is called an \FOtwo{} sentence. An \FOtwo{} sentence with only universal quantifiers is called a \UFOtwo{} sentence. An \FOtwo{} sentence with \emph{counting quantifiers} $\exists_{=k}$, $\exists_{\le k}$ and $\exists_{\ge k}$ is called a \Ctwo{} sentence, restricting that the number of assignments of the quantified variable satisfying the subsequent formula is exactly $k$, at most $k$ or at least $k$, respectively. A sentence can be possibly augmented with \emph{cardinality constraints}, which are expressions of the form of $|P| \bowtie k$ where $P$ is a predicate and $\bowtie$ is a comparison operator $\{<, \le, =, \ge, >\}$. We view cardinality constraints as atomic formulas that are satisfied if the number of true ground atoms of $P$ in a possible world meets the constraint.

\subsection{Weighted First-Order Model Counting}

The \textit{weighted first-order model counting (WFOMC)} problem takes the input consisting of a first-order sentence $\sentence$, a domain $\domain$ of size $n$, and a pair of weighting functions $(\weight, \negweight)$ that both map $\mathcal{P}$ to real weights.
Given a set $L$ of literals whose relations are in $\sentence$, the weight of $L$ is defined as
\begin{equation*}
W(L, \weight, \negweight):= \prod_{l \in L_T}\weight(\pred{l}) \cdot \prod_{l \in L_F}\negweight(\pred{l}),
\end{equation*}
where $L_T$ (resp. $L_F$) denotes the set of positive (resp. negative) literals in $L$, and $\pred{l}$ maps a literal $l$ to its corresponding relation name. We omit the symbols $\weight, \negweight$ and write $W(L)$ in short when the weighting functions are clear in the text.

\begin{example}\label{exmp:model-weight}
Let $\mathcal{P} = \{R, S\}$. Consider
%the sentence $\sentence = \forall x \forall y (R(x) \lor S(x,y))$ and
the weighting functions $\weight(R) = 2,\ \weight(S) = 3,\ \negweight(R) = \negweight(S) = 1$.
The weight of the literal set
\begin{equation*}
L = \{R(1), \lnot R(2), S(1,1), \lnot S(1,2), S(2,1), S(2,2)\}
\end{equation*}
is
\begin{equation*}
  \weight(R) \cdot \left(w(S)\right)^3 \cdot \negweight(R) \cdot \negweight(S) = 54.
\end{equation*}
\end{example}

\begin{definition}[Weighted First Order Model Counting]\label{def:wfomc}
    The WFOMC of a first-order sentence $\sentence$ over a finite domain $\domain$ under weighting functions $\weight, \negweight$ is
    \begin{equation*}
        \wfomc(\sentence, \domain, \weight, \negweight) := \sum_{\mu \in \fomodels{\sentence}{\domain}} W(\mu, \weight, \negweight).
    \end{equation*}
\end{definition}

% \lucien{slightly reworded the remark below}
We remark that as the weighting functions are defined in terms of relations, all positive ground literals of the same relation get the same weights, and so do all negative ground literals of the same relation, i.e., the weights are \emph{symmetric} w.r.t. domain elements.
%Thus, \cref{def:wfomc} is also called \emph{symmetric} weighted first-order model counting.
The symmetries of weights are crucial for efficient model counting as we shown later.
For asymmetric WFOMC where each ground literal can have its own weight, we refer the reader to \citet{WFOMC-FO3}, where it was shown that asymmetric WFOMC is \class{\#P}-hard even for \UFOtwo{} sentences.

\begin{example}
Consider the sentence $\sentence = \forall x \forall y: R(x) \lor S(x,y)$ and the weighting functions in \Cref{exmp:model-weight} over the domain $\domain = \{1, 2, \cdots, n\}$. Then
\begin{equation*}
\wfomc(\sentence, \domain, \weight, \negweight) = \left(2^{2n+1} + 3^n\right)^n.
\end{equation*}

In fact, for each domain element $i \in \domain$, either $R(i)$ is true $S(i,j)$ is not limited for any $j \in \domain$ which contributes weight $2 \cdot (3+1)^n$, or $R(i)$ is false and $S(i,j)$ is true for all $j \in \domain$ which contributes weight $3^n$. Multiplying the contributed weight of each element, we get the above value.
\end{example}

\subsection{1-Types and 2-Tables}

We will need the following notions to describe the algorithms.

A set of literals is \emph{maximally consistent} if it does not contain both an atom and its negation at the same time, and cannot be extended to a larger consistent set.

\begin{definition}[1-Type]
A \emph{1-type} of a first-order sentence $\sentence$ is a maximally consistent set of literals formed from predicates in $\sentence$ where each literal contains a single variable $x$.
\end{definition}

\begin{definition}[2-Table]
A \emph{2-table} of a first-order sentence $\sentence$ is a maximally consistent set of literals formed from predicates in $\sentence$ where each literal uses both the variables $x, y$.
\end{definition}

We alternatively write a 1-type (and a 2-table) as a conjunction of literals that it contains.
For example, $\sentence = \forall x \forall y: S(x,y) \lor R(y)$ has four 1-types: $S(x,x) \land R(x)$, $S(x,x) \land \lnot R(x)$, $\lnot S(x,x) \land R(x)$ and $\lnot S(x,x) \land \lnot R(x)$, and four 2-tables: $S(x,y) \land S(y,x)$, $S(x,y) \land \lnot S(y,x)$, $\lnot S(x,y) \land S(y,x)$ and $\lnot S(x,y) \land \lnot S(y,x)$. Intuitively, a 1-type interprets unary and reflexive binary relations for a single domain element, and a 2-table interprets binary relations for a pair of distinct domain elements.

\begin{definition}[1-Type Configuration]
Let $C = \{C_1, C_2, \cdots, C_p\}$ be the set of possible 1-types of a sentence $\sentence$.\footnote{We call a 1-type possible in $\sentence$ if it can be instantiated in some model of $\sentence$, which was also called valid in literatures~\cite{WFOMC-FO2-faster,Complexity-C2}.} A \emph{1-type configuration} w.r.t. $\sentence$ is a vector $\otc(\domain) = (\otce_1, \cdots, \otce_p)$ indicating the number of elements partitioned to each 1-type over the domain $\domain$.
\end{definition}

We often omit $\domain$ for simplicity if it is clear in the context.

%We also need a \emph{2-table configuration} to indicate the number of each 2-table between $e$ and elements with each 1-type in $\domain$ excluding $e$. A 2-table configuration of an element $e$ under a 1-type configuration $\otc(\domain)$ ($e \not \in \domain$) is a vector of 1-type configurations $\ttc^e = (\ttc^e_1, \cdots, \ttc^e_U)$ where $\ttc^e_i = (\ttce^e_{i,1}, \cdots, \ttce^e_{i,B})$ satisfying that $\sum_{j=1}^B \ttce^e_{i,j} = \otce_i$.

\subsection{Data Complexity of WFOMC}\label{sec:data-complexity}

We consider the \emph{data complexity}, which measures the runtime of WFOMC in terms of the size of the domain, regarding the sentence and the weighting functions as fixed. The fragments whose WFOMC have polynomial time data complexity are called \emph{domain-liftable}.

The fragment \UFOtwo{} was shown to be domain-liftable by \citet{WFOMC-UFO2} and \citet{WFOMC-FO3}.
%Later \citet{}
We briefly introduce the latter algorithm whose notations will be used in the rest of the paper.

Consider $\wfomc(\sentence, \domain, \weight, \negweight)$ for a \UFOtwo{} sentence in the \emph{prenex normal form} $\sentence = \forall x \forall y \ \psi(x,y)$ where $\psi(x,y)$ is a quantifier-free \FOtwo{} sentence. Suppose that $\domain = \{1, 2, \cdots, n\}$. Let $\refformula{\psi}(x,y)=\psi(x,y)\land \psi(y,x)$.
Then $\sentence$ can be expanded as conjunction of ground formulae over $\domain$:
\begin{equation}\label{eq:D}
  \begin{aligned}
    \sentence = \left( \bigwedge_{i \in \domain} \psi(i,i) \right) \land \left( \bigwedge_{i,j \in \domain: i<j} \refformula{\psi}(i,j) \right).
  \end{aligned}
\end{equation}

Let $C = \{C_1, C_2, \cdots, C_p\}$ be the set of possible 1-types of $\sentence$. Suppose that the 1-type of element $i$ is determined as $\tau_i$ ($\tau_i \in C$). After substituting unary and reflexive binary literals in $\refformula{\psi}(i,j)$ with true or false according to $\tau_i$ and $\tau_j$, the formula for a pair of elements $(i,j)$ does not have common ground literals with any other pair, hence the 2-tables between each pair of elements can be selected independently.
Let $D$ be the set of 2-tables of $\sentence$, and define $D_{s,t} = \{\pi\in D: s(a) \land t(b) \land \pi(a,b) \models \refformula{\psi}(a,b)\}$ for each $s, t \in C$.
The WFOMC can be computed as:
\begin{equation}\label{eq:basicwfomc1}
  \begin{aligned}
    % & \wfomc(\sentence, \domain, \weight, \negweight) = \\
    % &\qquad
    \sum_{\tau_1, \dots, \tau_n\in C} \ \prod_{i \in \domain} W(\tau_i) \prod_{i,j \in \domain: i<j} r_{\tau_i,\tau_j},
  \end{aligned}
\end{equation}
where
\begin{equation}\label{eq:def_r}
  r_{s, t} = \sum_{\pi \in D_{s, t}} W(\pi).
\end{equation}

Furthermore, from \cref{eq:basicwfomc1} we know that the computation of $\wfomc(\sentence, \domain, \weight, \negweight)$ only depends on the numbers of elements assigned to each 1-type. Therefore, we compute \cref{eq:basicwfomc1} in time polynomial in $n$ by enumerating the 1-type configuration instead of the 1-types for each element:
\begin{equation}\label{eq:basicwfomc2}
  \begin{aligned}
    \wfomc(\sentence, \domain, \weight, \negweight) = \sum_{\otce_1+\cdots+\otce_p=n} \Bigg(\frac{n!}{\otce_1! \cdots \otce_p!}
      \cdot \prod_{i=1}^p (W(C_i))^{\otce_i} \left(r_{i,i}\right)^{\binom{\otce_i}{2}} \prod_{1 \le i < j \le p} \left(r_{i,j}\right)^{\otce_i\otce_j}\Bigg).
  \end{aligned}
\end{equation}

\noindent Recently, \citet{DBLP:conf/aaai/TothK23} proposed a dynamic programming version of this approach, which partially inspires our dynamic programming algorithm.

\subsection{Evidence}
\label{subsec:evidence}

In real-world applications, we often care about WFOMC conditioned on some set of ground literals, called \emph{evidence}.
Evidence can be in either the \emph{open-world} or the \emph{closed-world} form.
% The \emph{evidence} specifies the truth of ground atoms. We use evidence in both the \emph{closed-world} form and the \emph{open-world} form.

\begin{definition}[Open-World Evidence]
The \emph{open-world evidence} is a consistent set of ground literals $\binevidence$, where the truth of each ground atom is restricted by its positivity or negativity in $\binevidence$.
Ground atoms not in $\binevidence$ are considered unknown (i.e., the truth is not restricted).
\end{definition}

\begin{definition}[Closed-World Evidence]
The \emph{closed-world evidence} of a set of predicates $\mathcal{P'} \subseteq \mathcal{P}$ is a set of ground atoms $\binevidence$ formed from $\mathcal{P'}$.
Ground atoms of the predicates in $\mathcal{P'}$ appearing in $\binevidence$ are supposed to be true, while those not appearing in $\binevidence$ are supposed to be false.
\end{definition}

\begin{remark}
  The closed-world notion differs from the closed-world assumption in probabilistic databases~\cite{probabilistic-databases}, where all ground atoms not appearing in the database are assumed to be false. In our case, only the ground atoms of the predicates in the evidence are assumed to be false, while those of other predicates are not restricted.
\end{remark}

The two forms of the evidence can be transformed to each other. Transforming the closed-world representation to the open-world representation is straightforward by adding the default negative ground literals to the evidence set. For instance, the closed-world evidence $\binevidence = \{R(1,2), R(2,2)\}$ of the predicate set $\mathcal{P'}=\{R\}$ over the domain $\domain = \{1,2\}$ corresponds to the open-world evidence $\binevidence' = \{\lnot R(1,1), R(1,2), \lnot R(2,1), R(2,2)\}$.
On the other hand, for example, consider the open-world evidence $\binevidence = \{R(1,2), R(2,3), \lnot R(3,4)\}$.
%indicates that $R(1,2)=R(2,3)=\top$, $R(3,4)=\bot$, and other ground atoms are unknown.
To transform $\binevidence$ to the closed-world form, we introduce two fresh binary predicates $R^{\top}$ and $R^{\bot}$, add the sentence $\forall x \forall y: (R^{\top}(x,y) \to R(x,y)) \land (R^{\bot}(x,y) \to \lnot R(x,y))$ and write $\binevidence' = \{R^{\top}(1,2), R^{\top}(2,3), R^{\bot}(3,4)\}$ which is in the closed-world form.

The data complexity of WFOMC is extended to the case with evidence, where the evidence is treated as another input (besides the domain size) of the problem.
Note that the number of ground literals in the evidence is polynomial in the domain size.
Given a class $\mathfrak{E}$ of evidence, we say that a first-order fragment $\mathfrak{L}$ with $\mathfrak{E}$ is \emph{domain-liftable} if the WFOMC conditioned on any evidence in $\mathfrak{E}$, i.e., WFOMC of $\sentence \land \binevidence$ for every $\sentence \in \mathfrak{L}$ and every $\binevidence \in \mathfrak{E}$, can be computed in time polynomial in the domain size.

\emph{Unary} and \emph{binary} evidence are two particular classes of evidence, where unary evidence consists of ground literals (or atoms) containing a \emph{single} element in the domain, and binary evidence consists of ground literals (or atoms) containing two \emph{distinct} elements.
For example, unary evidence can contain $R(1,1)$ though $R$ is a binary predicate.

In this paper, we push the boundaries of data-liftability to the case where the underlying Gaifman graph of the binary evidence has bounded treewidth.

\begin{definition}[Gaifman Graph]
The \emph{Gaifman graph} with respect to the domain $\domain$ and the closed-world binary evidence $\binevidence$ is an undirected graph $G_{\domain,\binevidence}$ of $n$ vertices, in which there is an edge between $a$ and $b$ if and only if $\binevidence$ contains $P(a,b)$ or $P(b,a)$ for some $P \in \mathcal{P}$.
\end{definition}

\Cref{fig:example-tree-decomp} shows an example of the Gaifman graph of the closed-world binary evidence $\{R(a,b), \\ R(a,c), R(b,c), R(c,d), R(c,e), R(d,g), R(e,g)\}$.

\begin{figure}[tbp]
  \centering
  \includegraphics[width=0.4\textwidth]{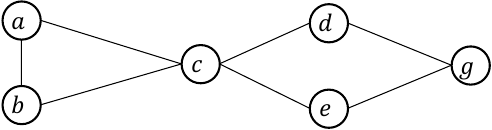}
  \includegraphics[width=0.47\textwidth]{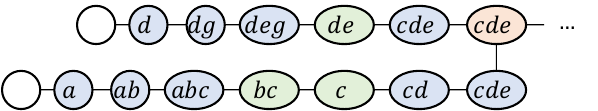}
  \caption{An example of the Gaifman graph (left) of closed-world binary evidence and its nice tree decomposition (right). The leaf, introduce, forget and join nodes are colored in white, blue, green and orange respectively. The last forget nodes are omitted for simplicity.}
  \label{fig:example-tree-decomp}
\end{figure}

For ease of presentation, we will alternatively write the evidence $\binevidence$ as a conjunction of ground literals according to its semantics, i.e., $\binevidence = \bigwedge_{l \in \binevidence} l$ for open-world evidence and $\binevidence = \bigwedge_{l \in \binevidence} l \land \bigwedge_{l \notin \binevidence: \mathsf{pred}(l) \in \mathcal{P}'} \lnot l$ for closed-world evidence, where $\mathsf{pred}(l)$ maps a literal $l$ to its corresponding relation name.
% \lucien{@Qipeng: please check this aligns with our approach.}

% \begin{definition}[Open-World Evidence]
% The \emph{open-world unary evidence} is a set of ground literals $\uevidence$ of unary predicates. The \emph{open-world binary evidence} is a set of ground literals $\binevidence$ of binary predicates. Ground atoms not appearing in the open-world evidence are not restricted.
% \end{definition}

% \begin{remark}
% The two forms of the evidence can be transformed to each other. Transforming the closed-world representation to the open-world representation is straightforward by adding the default negative ground literals to the evidence set. For instance, the closed-world binary evidence $\binevidence = \{R(1,2), R(2,2)\}$ of the predicate set $\mathcal{P'}=\{R\}$ over the domain $\domain = \{1,2\}$ corresponds to the open-world binary evidence $\binevidence' = \{\lnot R(1,1), R(1,2), \lnot R(2,1), R(2,2)\}$.

% On the other hand, for example, the open-world binary evidence $\binevidence = \{R(1,2), R(2,3), \lnot R(3,4)\}$ indicates that $R(1,2)=R(2,3)=\top$, $R(3,4)=\bot$, and other ground atoms are unknown. To transform $\binevidence$ to the closed-world form, we can introduce two fresh binary predicates $R^{\top}$ and $R^{\bot}$, add the sentence
% \begin{equation*}
%   \forall x \forall y \left( (R^{\top}(x,y) \to R(x,y)) \land (R^{\bot}(x,y) \to \lnot R(x,y)) \right)
% \end{equation*}
% and write $\binevidence' = \{R^{\top}(1,2), R^{\top}(2,3), R^{\bot}(3,4)\}$ which is the closed-world binary evidence.
% \end{remark}

\subsection{Treewidth and Tree Decompositions}

We refer to the common definition of tree decomposition (e.g., \nobracketcite{param-complexity}).

\begin{definition}[Tree Decomposition]\label{def:tree-decomp}
A \emph{tree decomposition} of a graph $G(V_G, E_G)$ is a tree $T(V_T, E_T)$ where each node $u \in V_T$ holds a \emph{bag} $B_u \subseteq V_G$ and the tree satisfies the following conditions:
\begin{itemize}
    \item $\bigcup_{u \in V_T} B_u = V_G$.
	\item For each edge $(a, b) \in E_G$, there is a tree node $u \in V_T$ such that $a, b \in B_u$.
	\item For each vertex $a \in V_G$, the nodes $u$'s with $a \in B_u$ form a connected component in $T$.
\end{itemize}
\end{definition}

The \emph{width} of a tree decomposition is $\max_{u \in T}\{|B_u|\} - 1$.
The \emph{treewidth} of a graph $G$ is the lowest width among its possible tree decompositions.

The following form of tree decomposition will be used to simplify the description of the algorithm.

\begin{definition}[Nice Tree Decomposition]\label{def:nice-tree-decomp}
For a graph of treewidth $k$, its tree decomposition $T(V_T, E_T)$ of width $k$ is \emph{nice} if a root node $root \in V_T$ can be chosen such that the tree rooted at $root$ satisfies:
\begin{itemize}
  \item Each node has at most two children.
  \item $B_{root} = \emptyset$ and $B_l = \emptyset$ for every leaf $l$.
  \item Each non-leaf node is of one of the following three types:
  \begin{itemize}
    \item An \emph{introduce node} $u$ has exactly one child $v$ such that $B_u = B_v \cup \{a\}$ for some $a \notin B_v$.
    \item A \emph{forget node} $u$ has exactly one child $v$ such that $B_u = B_v \setminus \{a\}$ for some $a \in B_v$.
    \item A \emph{join node} $u$ has exactly two children $v_1, v_2$ such that $B_u = B_{v_1} = B_{v_2}$.
  \end{itemize}
\end{itemize}
\end{definition}

An example of the nice tree decomposition is shown in \Cref{fig:example-tree-decomp}, where the treewidth is 2.
Together with \citeauthor{Bodlaender93}’s work \shortcite{Bodlaender93}, it was shown that such a tree decomposition can be computed efficiently.

\begin{lemma}{\emph{\cite{param-complexity,Bodlaender93}}}
\label{lemma:tree-decomp}
For a graph $G(V_G, E_G)$ of treewidth $k$, its nice tree decomposition of size $O(k\cdot |V_G|)$ can be computed in time linear in $|V_G|$.
\end{lemma}

% We will make constraints on the treewidth of the following graph of binary evidence.

% \begin{definition}[Gaifman Graph]
% The \emph{Gaifman graph} with respect to the domain $\domain$ and the binary evidence $\binevidence$ is an undirected graph $G_{\domain,\binevidence}$ of $n$ vertices, in which there is an edge between $a$ and $b$ if and only if $\binevidence$ contains $P(a,b)$ or $P(b,a)$ for some $P \in \mathcal{P}$.
% \end{definition}

\subsection{Markov Logic Networks}
Markov Logic Networks (MLNs) \citep{richardson:mlns} is a model used in the area of statistical relational learning.
We will use MLNs to encode the Watts-Strogatz model later in the experiments.

An MLN $\Phi$ is a set of weighted quantifier-free first-order logic formulas $\alpha_1, \cdots, \alpha_k$ with weights $w_1, \cdots, w_k$ taking on values from the real domain or infinity:
\begin{equation*}
  \Phi = \{(w_1, \alpha_1), (w_2, \alpha_2), \cdots, (w_k, \alpha_k)\}.
\end{equation*}

Given a domain $\Delta$, the MLN defines a probability distribution over possible worlds such that
\begin{equation*}
P_{\Phi, \Delta}(\omega) = \begin{cases}
\frac{\exp{\left(\sum_{\substack{(w_i, \alpha_i)\in \Phi_\mathbb{R}}}w_i \cdot N(\alpha_i, \omega)\right)}}{Z}, & \text{if } \omega \models \Phi_{\infty} \\
0, & \text{otherwise}
\end{cases}
\end{equation*}
where $\Phi_\mathbb{R}$ denote the real-valued (soft) weight-formula pairs, $\Phi_{\infty}$ the $\infty$-valued (hard) weight-formula pairs, $Z$ is the normalization constant ensuring valid probability values, and $N(\alpha_i, \omega)$ is the number of substitutions for variables in $\alpha_i$ by constants that produce a ground formula satisfied in the world $\omega$.
The distribution formula is equivalent to the one of a Markov Random Field \citep{koller:mrf}. Hence, an MLN, along with a domain, defines a probabilistic graphical model, and inference in the MLN is thus inference over that model.

Inference (and also learning) in MLNs is reducible to WFOMC \cite{WFOMC-FO2}.
For each $(w_i, \alpha_i(\mathbf{x}_i)) \in \Phi_{\mathbb{R}}$, where $\mathbf{x}_i$ is the free variables, introduce a new formula
$\forall \mathbf{x}_i: \xi_i(\mathbf{x}_i) \leftrightarrow \alpha_i(\mathbf{x}_i)$,
where $\xi_i$ is a fresh predicate, $w(\xi_i) = \exp(w_i)$ and $\overline{w}(\xi_i) = 1$. Let $w(Q) = \overline{w}(Q) = 1$ for all other
predicates $Q$.
Consider sentence $\Gamma$ created by the above procedure, with hard formulas added to $\Gamma$ as they are.
We can then compute the probability of a query $\phi$ as
$P_{\Phi, \Delta} (\phi) = \frac{\wfomc(\Gamma \wedge \phi, \Delta, w, \overline{w})}{\wfomc(\Gamma, \Delta, w, \overline{w})},$
where $\phi$ is a set of constraints that can be both hard and soft.

\section{Approach}

We describe a novel algorithm that given a sentence $\sentence$ in the two-variable fragment, a finite domain $\domain$ of size $n$, the weighting functions $\weight, \negweight$, the \emph{closed-world} unary evidence $\uevidence$ and binary evidence $\binevidence$ whose Gaifman graph $G_{\domain,\binevidence}$ has treewidth $k$ for some constant $k$ (which we refer to as \emph{bounded-treewidth binary evidence}), it computes $\wfomc(\sentence \land \uevidence \land \binevidence, \domain, \weight, \negweight)$ in time polynomial in $n$.
The open-world evidence is handled by transforming it to the closed-world form as described in \Cref{subsec:evidence}.

%We make further assumptions on the problem.
We start by providing an algorithm for the universally quantified two-variable fragment \UFOtwo{}. At the end of the section, we use existing reduction techniques to extend the result to \FOtwo{} and \Ctwo{} with cardinality constraints.

We assume that the nice tree decomposition $T(V_T, E_T)$ of the Gaifman graph $G_{\domain,\binevidence}$ is given, which has a root $root \in V_T$ and size $O(kn)$.
By \Cref{lemma:tree-decomp}, such a tree decomposition can be computed in time linear in $n$.
% we can compute a nice tree decomposition $T(V_T, E_T)$ of size $O(kn)$ with a root $root \in V_T$ for the Gaifman graph $G_{n,\binevidence}$ in time linear in $n$.

Let $C = \{C_1, C_2, \cdots, C_p\}$ be the set of possible 1-types of $\sentence$, and let $D$ be the set of 2-tables of $\sentence$. Define $D_{s,t}$ for each $s,t \in C$ as in \Cref{sec:data-complexity}.
%Suppose that the 1-type of element $i$ is determined as $\tau_i$ ($\tau_i \in C$).
% We write $\binevidence[\domain']$ for the partial binary evidence $\binevidence$ with respect to a subset $\domain' \subseteq \domain$, which is the subset of $\binevidence$ that only contains ground atoms with constants in $\domain'$.
When we talk about the satisfaction of evidence with respect to a subset $\domain'\subseteq\domain$, we only consider the truth of ground atoms that only contain constants in $\domain'$.

\subsection{The Recursion Framework}

% \todo{examples}

In this subsection, we present the algorithm for the \UFOtwo{} fragment with unary evidence $\uevidence$ and binary evidence $\binevidence$.

For a node $u \in V_T$, $B_u\subseteq \domain$ is the set of elements in the bag of $u$. Let $B^*_u$ denote the union of $B_v$ for all nodes $v$ in the subtree rooted at $u$, and let $S_u$ denote $B^*_u \setminus B_u$.
Given the 1-types $\tau = \bigwedge_{a \in B_u} \tau_a$ of elements in $B_u$ satisfying $\uevidence$
%the 2-tables $\pi = \bigwedge_{a,b \in B_u} \pi_{a,b}$ between elements in $B_u$,
and the 1-type configuration $\otc = \otc(S_u)$ for elements in $S_u$, define a \emph{partial model} of $u$ as a set of ground literals consisting of the followings:
\begin{itemize}
  \item The model of $\sentence$ over $S_u$ satisfying $\uevidence$ and $\binevidence$ and consistent with $\otc$ (i.e., there are $\otce_i$ elements having the 1-type $C_i$ for each $1\le i\le p$), and
  \item The 2-tables between $B_u$ and $S_u$ satisfying $\binevidence$.
\end{itemize}

Let $\mathcal{L}_{u, \tau, \otc}$ be the set of possible partial models of $u$ with respect to $\tau, \otc$.
Define $f(u,\tau,\otc)$ as the sum of weights of partial models in $\mathcal{L}_{u, \tau, \otc}$, i.e., $f(u,\tau,\otc) = \sum_{L \in \mathcal{L}_{u, \tau, \otc}} W(L)$.
Then our target is to compute
\begin{equation*}
  \wfomc(\sentence \land \uevidence \land \binevidence, \domain, \weight, \negweight) = \sum_{\otce_1 + \cdots + \otce_p = n} f(root, \top, \otc).
\end{equation*}
This is because $B_{root} = \emptyset$ and $S_{root} = \domain$, in which it holds that $\bigcup_{\otce_1 + \cdots + \otce_p = n} \mathcal{L}_{root, \top, \otc} = \fomodels{\sentence \land \uevidence \land \binevidence}{\domain}$.

Next we compute $f(u,\tau,\otc)$ by recursion on the tree decomposition. The computation depends on the type of $u$.
Recall that in \Cref{sec:data-complexity} we write $\sentence$ as $\forall x \forall y \ \psi(x,y)$ for some quantifier-free formula $\psi$ and write $\refformula{\psi}(x,y)=\psi(x,y)\land\psi(y,x)$. As $\binevidence$ is in the close-world representation w.r.t. a set of predicates $\mathcal{P}'$, all atoms of $\mathcal{P'}$ not in $\binevidence$ are supposed to be false. We refine the definition of $r_{s,t}$ in \cref{eq:def_r} to align with $\binevidence$ for element pairs not in $\binevidence$, i.e.,
%as the weighted sum of valid 2-tables between two elements in 1-types $s$ and $t$ respectively, in addition that all binary ground literals of $\mathcal{P'}$ appear negatively, i.e.,
\begin{equation*}
  r_{s,t} = \sum_{\substack{\pi \in D_{s,t}:\ \pi(a,b) \models \bigwedge_{\mathsf{p} \in \mathcal{P'}} \lnot \mathsf{p}(a,b) \land \lnot \mathsf{p}(b,a)}} W(\pi)
\end{equation*}
%\lucien{closed world or open world?}
%\kqp{Closed-world, as stated in the beginning of this section.}
for every pair of 1-types $s,t \in C$.
The original definition of $r_{s,t}$ in \cref{eq:def_r} can be regarded as in the case of $\mathcal{P'} = \emptyset$.

\noindentparagraph{The Leaf Node}
If $u$ is a leaf, we have $B_u = \emptyset$ by \Cref{def:nice-tree-decomp}. In this case, the only valid pair of $(\tau,\otc)$ is $(\top,\bm{0})$. We assign $1$ to the $f$ value:
% as an initial value:
\begin{equation}\label{eq:f-leaf}
  f(u,\top,\bm{0}) = 1.
\end{equation}

\noindentparagraph{The Introduce Node}
Let $v$ be the child node of $u$, and $a$ be the unique element in $B_u \setminus B_v$. Each partial model $L \in \mathcal{L}_{u,\tau,\otc}$ can be decomposed to the following three parts:
\begin{itemize}
  \item A model $\mu \in \fomodels{\sentence \land \uevidence \land \binevidence}{S_u}$ consistent with $\otc$;
  \item The 2-tables $\pi_1$ between $B_u \setminus \{a\}$ and $S_u$;
  \item The 2-tables $\pi_2$ between $a$ and $S_u$.
\end{itemize}

Let $\tau_a$ be the 1-type of $a$. Since $S_u = S_v$, $\mu \land \pi_1$ is a partial model $L' \in \mathcal{L}_{v, \ \tau \setminus \tau_a,\ \otc}$ where $\tau \setminus \tau_a$ means $\tau$ removing $\tau_a$. Therefore, each $L \in \mathcal{L}_{u,\tau,\otc}$ where $a$ has 1-type $\tau_a$ can be regarded as augmenting $\pi_2$ to a partial model in $\mathcal{L}_{v, \ \tau \setminus \tau_a,\ \otc}$. We then have
% \lucien{Why isn't $\otc$ updated here?}
% \kqp{Because $S_u = S_v$.}
\begin{equation}\label{eq:f-forget}
\begin{aligned}
  f(u,\tau,\otc)
  = \left( \prod_{i=1}^p \left(r_{\tau_a, C_i}\right)^{\otce_i} \right) \cdot f(v,\ \tau \setminus \tau_a,\ \otc).
\end{aligned}
\end{equation}
Here, we use the property that $a$ does not connect to any element in $S_u$ in the Gaifman graph, i.e., all ground atoms containing $a$ and elements in $S_u$ should be false in the evidence $\binevidence$.
For example, the element $d$ in the introduce node $cd$ in \Cref{fig:example-tree-decomp} does not connect to any element in $S_{cd} = \{a,b\}$, hence $R(a,d), R(d,a), R(b,d), R(d,b)$ are all false in the evidence.
%where note that the weighted sum of 2-tables between $a$ and elements in $S_u$ is computed using the same idea as in \cref{eq:basicwfomc2}.

\noindentparagraph{The Forget Node}
Let $v$ be the child node of $u$, and $a$ be the unique element in $B_v \setminus B_u$. Note that $S_u = S_v \cup \{a\}$.
%for $u$, element $a$ is in $S_u$ and for $v$, element $a$ is in $B_v$.
Each partial model $L \in \mathcal{L}_{u,\tau,\otc}$ can be decomposed to the following five parts:
\begin{itemize}
  \item The 1-type $\tau_a$ of $a$ satisfying $\uevidence$;
  \item A model $\mu \in \fomodels{\sentence \land \uevidence \land \binevidence}{S_v}$ consistent with $\otc^{-\tau_a}$, where $\otc^{-\tau_a}$ means $\otc$ subtracting $1$ at the position of $\tau_a$;
  \item The 2-tables $\pi_1$ between $a$ and $S_v$;
  \item The 2-tables $\pi_2$ between $B_u$ and $a$;
  \item The 2-tables $\pi_3$ between $B_u$ and $S_v$.
\end{itemize}

As $B_v = B_u \cup \{a\}$, we know that $\mu \land \pi_1 \land \pi_3$ is a partial model $L' \in \mathcal{L}_{v, \ \tau \land \tau_a, \ \otc^{-\tau_a}}$. Therefore, each $L \in \mathcal{L}_{u,\tau,\otc}$ where $a$ has 1-type $\tau_a$ can be regarded as augmenting $\tau_a$ and $\pi_2$ to a partial model in $\mathcal{L}_{v, \ \tau \land \tau_a, \ \otc^{-\tau_a}}$. Let $B_u = \{b_1, \cdots, b_{m}\}$. We can then compute $f(u,\tau,\otc)$ as follows:
\begin{equation}\label{eq:f-introduce}
\begin{aligned}
  f(u,\tau,\otc) = \sum_{\tau_a\! \in C: \tau_a\! \models \uevidence} \ \sum_{\pi_{b_1}\! \in D_{a,b_1,\tau_a,\tau_{b_1}}}\! \cdots\! \sum_{\pi_{b_m}\! \in D_{a,b_m,\tau_a,\tau_{b_m}}}
   \left( W(\tau_a) \cdot \prod_{b \in B_u} W(\pi_{b}) \right) \cdot f(v, \tau \land \tau_a, \ \otc^{-\tau_a}),
\end{aligned}
\end{equation}
where $D_{a,b,s,t}$ is the set of the 2-tables between the elements $a$ and $b$ having the 1-types $s$ and $t$, respectively that satisfy $\binevidence$, i.e.,
% \lucien{do we need to consider the atoms not in $\binevidence$?}
% \kqp{Yes}
$D_{a,b,s,t} = \{\pi \in D_{s,t}: \pi(a,b)\models \binevidence\}$.

\noindentparagraph{The Join Node}
Let $v_1, v_2$ be the children nodes of $u$. Note that $S_u = S_{v_1} \cup S_{v_2}$. For each $L \in \mathcal{L}_{u,\tau,\otc}$, it can be decomposed to the following three parts:
\begin{itemize}
  \item A partial model $L_1 \in \mathcal{L}_{v_1, \tau, \bm{z_1}}$;
  \item A partial model $L_2 \in \mathcal{L}_{v_2, \tau, \bm{z_2}}$ where $(\bm{z_1})_i+(\bm{z_2})_i = \otce_i$ for each $1 \le i \le p$;
  \item The 2-tables between $S_{v_1}$ and $S_{v_2}$.
\end{itemize}

Therefore, we have
\begin{equation}\label{eq:f-join}
\begin{aligned}
  f(u,\tau,\otc) = \sum_{\bm{z_1}+\bm{z_2} = \otc} f(v_1, \tau,\bm{z_1}) \cdot f(v_2, \tau,\bm{z_2})
   \cdot \prod_{i=1}^p \prod_{j=1}^p \left(r_{C_i, C_j}\right)^{(\bm{z_1})_i \cdot (\bm{z_2})_j},
\end{aligned}
\end{equation}
where $\bm{z_1}+\bm{z_2}$ denotes the element-wise addition of the two vectors.
% It is the same idea used in \cref{eq:basicwfomc2} to compute the weighted sum of 2-tables between $S_{v_1}$ and $S_{v_2}$.
In the above equation, we use the same idea as handing the introduce node to compute the weights of 2-tables between $S_{v_1}$ and $S_{v_2}$.

\begin{remark}
In the computation of \cref{eq:f-leaf,eq:f-introduce,eq:f-forget,eq:f-join} we might encounter invalid tuples of $(u,\tau,\otc)$, e.g., $\tau$ not satisfying $\uevidence$ or the tuples where $\sum_{i=1}^p \otce_i \neq |S_u|$. The $f$-values of these tuples are set to $0$ and thus do not affect the computation with valid tuples.
\end{remark}

\subsection{Time Complexity}

We show that computing the recursion of $f$ by \cref{eq:f-leaf,eq:f-introduce,eq:f-forget,eq:f-join} takes time polynomial in $n$. Let $p$ be the number of 1-types of $\sentence$ and $q$ be the number of 2-tables of $\sentence$. Recall that by \Cref{def:tree-decomp}, for every node $u \in V_T$, we have $|B_u| \le k+1$.

\begin{itemize}
  \item For a leaf node, it takes $O(1)$ time since there is only one valid tuple of $(\tau,\otc)$.
  \item For a non-leaf node, the number of valid tuples of $(\tau,\otc)$ is $O\left(p^{k+1} \cdot n^p\right)$. The time of computing \cref{eq:f-introduce,eq:f-forget,eq:f-join} is as follows:
      \begin{itemize}
        \item For an introduce node, computing the weights of 2-tables between $a$ and $S_u$ takes $O(p)$ time, so the total time on the node is $O\left(p^{k+2} \cdot n^p\right)$.
        \item For a forget node, enumerating $\tau_a$ and $\pi_{a,b}$ for all $b \in B_u$ takes $O\left(p \cdot q^{k+1}\right)$ time and the computation after the enumeration takes $O(k)$ time, so the total time on the node is $O\left(p^{k+2} \cdot q^{k+1} \cdot k \cdot n^p\right)$.
        \item For a join node, enumerating $\bm{z_1}, \bm{z_2}$ takes $O(n^p)$ time and computing the weights of 2-tables between $S_{v_1}$ and $S_{v_2}$ takes $O(p^2)$ times, so the total time on the node is $O\left(p^{k+3} \cdot n^{2p}\right)$.
      \end{itemize}
\end{itemize}

By \Cref{lemma:tree-decomp}, the number of nodes on a tree is $O(kn)$ and the tree can be computed in time $O(n)$. Therefore, the total time is $O(n + kn \cdot p^{k+3} \cdot q^{k+1} \cdot k \cdot n^{2p})$. As the sentence $\sentence$ is fixed when we consider data complexity, $p$ and $q$ are also fixed. We then obtain the following theorem.

\begin{theorem}\label{thm:ufo2}
\UFOtwo{} with unary evidence and bounded-treewidth binary evidence is domain-liftable.
\end{theorem}

\subsection{Implications of the Result}

Beside \UFOtwo{}, our algorithm also works for other domain-liftable two-variable fragments which have a \emph{modular transformation} to \UFOtwo{} defined by \citet{WFOMC-FO2}. The property says that for any sentence in such fragment, its WFOMC can be obtained from the WFOMC of a \UFOtwo{} sentence over any finite domain.
It was shown that there are modular transformations to \UFOtwo{} for \FOtwo{} \cite{WFOMC-FO2} and \Ctwo{} \cite[Theorem 4]{WFOMC-C2} possibly with cardinality constraints \cite[Proposition 5]{WFOMC-C2}, hence we obtain the following theorem:
%conclude that \FOtwo{} and \Ctwo{}, possibly with cardinality constraints and bounded-treewidth binary evidence are domain-liftable.

%\begin{definition}[Modular Transformation]
%The fragment $\mathcal L$ has a \emph{modular transformation} to \UFOtwo{} if the following holds:
%for any sentence $\sentence \in \mathcal L$ and any weighting functions $\weight, \negweight$, there is a \UFOtwo{} sentence $\sentence'$ and weighting functions $\weight', \negweight'$ such that for any finite domain $\domain$ and any sentence $\phi$, it holds that $\wfomc(\sentence \land \phi, \domain, \weight, \negweight) = \wfomc(\sentence' \land \phi, \domain, \weight', \negweight')$.
%\end{definition}

%\FOtwo{} and \Ctwo{} with only unary evidence are shown to be domain-liftable using the trick of transferring unary evidence to cardinality constraints \cite[Appendix A]{WFOMS-FO2-journal}. However, it is unclear whether this trick works with binary evidence, as the trick depends on the symmetry of elements while binary evidence breaks the symmetry.
%Fortunately, our algorithm incorporates with unary evidence $\uevidence$. For each node $u$, we only consider the 1-types $\tau$ of elements in $B_u$ such that $\tau$ satisfies $\uevidence$. By a simple induction from bottom to top on the tree, we claim that at every node $u$, only the weights of partial models satisfying $\uevidence$ will be counted.

%Combining these two implications, we obtain the following theorem.

\begin{theorem}
\FOtwo{} and \Ctwo{}, possibly with cardinality constraints, unary evidence and bounded-treewidth binary evidence are domain-liftable.
\end{theorem}

%Our algorithm could be also applied to asymmetric WFOMC problems, whose details we defer to the full version due to space constraints.

\subsection{Asymmetric Weights for Binary Predicates}

Recall that for symmetric weights, we compute the weight of a set of literals as $W(L, \weight, \negweight):= \prod_{l \in L_T}\weight(\pred{l}) \cdot \prod_{l \in L_F}\negweight(\pred{l})$.
Suppose that there is a set of literals $L_a$ with asymmetric weights $w_a: L_a \to \mathbb{R}$ and $\overline{w}_a: L_a \to \mathbb{R}$ that may differ from $\weight(\pred{l})$ or $\negweight(\pred{l})$.
We can now compute the weight of a set of literals as
\begin{equation*}
\begin{aligned}
    W(L, \weight, \negweight) = \prod_{l \not \in L_a, l \in L_T}\weight(\pred{l}) \cdot \prod_{l \not \in L_a, l \in L_F}\negweight(\pred{l}) \cdot \prod_{l \in L_a, l \in L_T}w_a(l) \cdot \prod_{l \in L_a, l \in L_F}\overline{w}_a(l).
\end{aligned}
\end{equation*}
To ensure that our algorithm works with this modification correctly, we need to add an edge $(i, j)$ to the Gaifman graph for each literal $P(i, j) \in L_a$ or $\lnot P(i,j) \in L_a$, if the edge does not exist originally.

Then we can use the same techniques as in the symmetric case to compute asymmetric WFOMC.
The data complexity is still polynomial in the domain size, as long as the Gaifman graph of $L_a$ has bounded treewidth.

\section{Experiments}
We evaluated the performance of our algorithm on several examples and compared it with two other model counting algorithms, d4 \cite{d4}, Forclift \cite{forclift}, and Crane2 \cite{crane}.
d4 is a propositional model counter that compiles a formula in conjunctive normal form to deterministic decomposable negation normal form (d-DNNF). It generally performs well when compared to other model counters, but it cannot leverage lifted methods.
Forclift is a model counter similarly to our algorithm that computes lifted inference and accepts binary evidence.
%last minute update
Crane2 is a model counter based on Forclift, which performs very well on a limited set of problems, but does not support full range of \FOtwo{} sentences yet.
For this reason, Crane2 is only used for friends and smokers experiments.
All experiments were run on a laptop with Intel i7-10510U CPU and 32GB of RAM.
\subsection{Friends and Smokers}

The first problem we chose is a variant of the often used \emph{friends and smokers} problem with the following sentence:
\begin{equation*}
  \begin{aligned}
    \forall x:&\; \lnot friends(x, x) \\
	\land \forall x \forall y:&\; friends(x, y) \rightarrow friends(y, x) \\
	\land \forall x \forall y:&\; (smokes(x) \wedge friends(x, y)) \rightarrow smokes(y).
  \end{aligned}
\end{equation*}

In our experiments, we provide open-world evidence that corresponds to the persons coming in cliques of 3 (i.e., due to the open-world nature of the evidence, they may also have other friends beyond the cliques), and suppose $|\Delta| = 3\cdot k$ ($k \geq 1$).
We compared the runtime of our algorithm (labeled TD-WFOMC), d4, Forclift, and Crane2.
Figure \ref{fig:friends} shows that our algorithm runs significantly faster than the other three, and scales well to large domain size.
%last minute update
We forgo the symmetry of friendships for Crane2 as symmetry is not supported.
In Figure \ref{fig:friends-vary-td}, we report the runtime of our algorithm for fixed domain size of 60 while varying the sizes of friend cliques.
%We also report the performance of our algorithm on the friends and smokers problem for a fixed domain size of 60 while varying the sizes of friend cliques, as shown in Figure \ref{fig:friends-vary-td}.

\begin{figure}[t]
\centering
\includegraphics[width=0.7\textwidth]{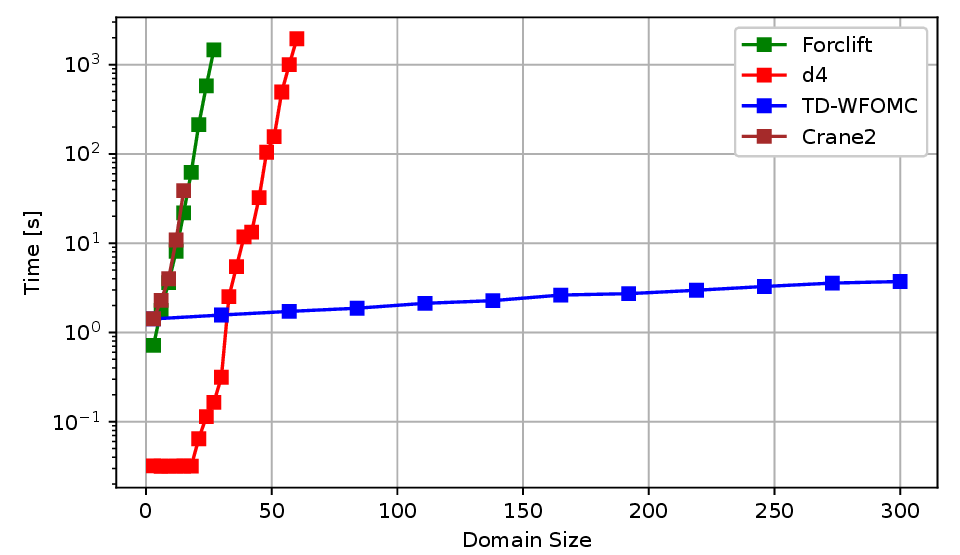} % Reduce the figure size so that it is slightly narrower than the column. Don't use precise values for figure width.This setup will avoid overfull boxes.
\caption{Time to compute the model count of friends and smokers problem with binary evidence}
\label{fig:friends}
%\end{figure}
\end{figure}

\begin{figure}[t]
\centering
\includegraphics[width=0.7\textwidth]{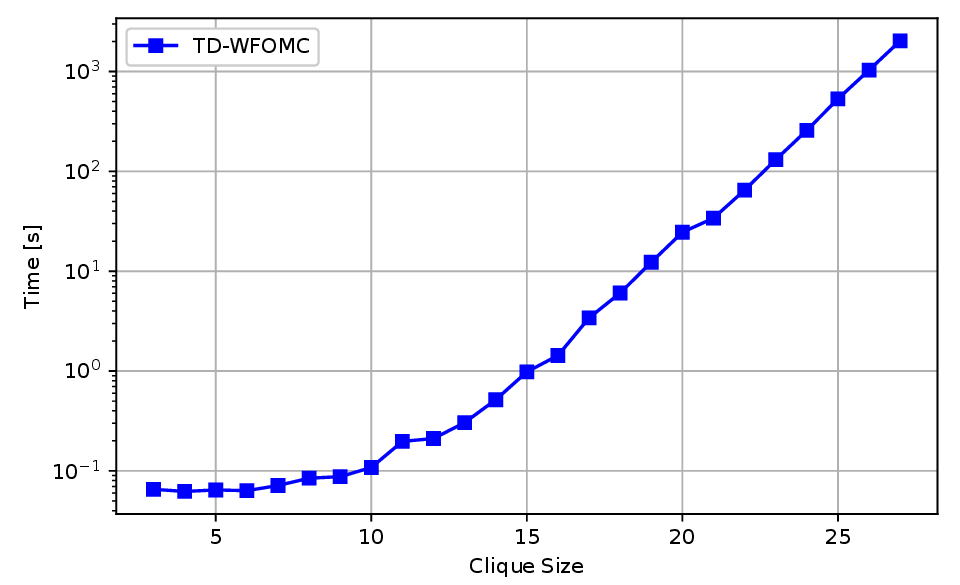} % Reduce the figure size so that it is slightly narrower than the column. Don't use precise values for figure width.This setup will avoid overfull boxes.
\caption{Time to compute the model count of friends and smokers problem with binary evidence, fixing $|\Delta| = 60$ for varying size of friend cliques.}
\label{fig:friends-vary-td}
%\end{figure}
\end{figure}

\subsection{Watts-Strogatz Model}
Another interesting problem we can tackle using WFOMC is performing inference on Watts-Strogatz (WS) model \cite{watts_collective_1998},
which is a random graph model widely used in network science to generate small-world networks.
The generation process of a random graph starts with a regular ring lattice with $N$ nodes, where each node is connected to exactly $K$ other nodes (assuming $K$ is an even integer), $K/2$ on each side, warping over the end.
To generate the random graph, each edge $(i, j)$
%such that $0 < (j - i) \mod N \le \frac{K}{2}$
is rewired with probability $\beta$ in some order, meaning that edge $(i, j)$ is replaced with edge $(i, k)$ ($i \neq k$) such that edge $(i, k)$ does not exist already.
The Watts-Strogatz model can be encoded as an MLN $\Psi$:
\begin{align*}
    &\infty, &\forall x : \neg WiredEdge(x, x)\\
    &\infty, &\forall x \forall y: \neg WiredEdge(x, y) \vee \neg WiredEdge(y, x)\\
    &\infty, &\forall x \exists_{=K/2} y: WiredEdge(x, y)\\
    %&\infty, &\forall x\forall y: \lnot WiredEdge(x,y) \vee \lnot WiredEdge(y, x) \\
    &w_1, &WiredEdge(x, y) \land \lnot EvidenceEdge(x ,y)\\
    &w_2, &WiredEdge(x, y) \land EvidenceEdge(x ,y)\\
    &\infty, & \forall x \forall y:  Edge(x ,y) \leftrightarrow \\
        &&(WiredEdge(x, y) \vee WiredEdge(y, x)),
\end{align*}
where $EvidenceEdge$ is interpreted as the starting regular ring lattice with binary evidence (in which each vertex has outgoing edges to only its right neighbors), $Edge(x, y)$ represents the rewired edges, and $w_1$ and $w_2$ depend on $\beta$.
We use wired edge to keep track of what edges are kept the same and what edges were changed, as well as to enforce that each node has degree $\ge K/2$, which holds in the Watts-Strogatz model.

Our algorithm can also handle random graphs created by this procedure regardless of what the starting graph is, as long as all the vertices share the same degree $d$ (which is also expected in the original model) and the starting graph has bounded treewidth, i.e., the Gaifman graph of binary evidence of $EvidenceEdge$ has bounded treewidth.
\iffalse
For a ring lattice $G(V_G, E_G)$ with $N$ nodes, each connected to $K$ neighbors, a tree decomposition $T(V_T, E_T)$ of treewidth $K$.
Let $T$ be a path graph $u_1, \cdots, u_{T - K}$ for all $u_i \in V_T$.
Let bag $B_{u_i}$ of vertex $u_i \in V_T$ be $\{v_1, \cdots, v_{K/2}, v_i, v_{i + 1}, \cdots, v_{i + K/2}\}$.
For each edge $(i, j) \in E_G$, it holds that $|i - j| \le K/2$.
Let $$ s = \begin{cases}
    min(i, j) & \text{if } i > K/2 \text{ and } j > K/2 \\
    i & \text{if } j \le K/2 < i\\
    j & \text{if } i \le K/2 < j\\
    1 & \text{otherwise}
\end{cases}$$
Take the bag $B_{u_{s}} = \{v_1, \cdots, v_{K/2}, v_s, v_{s + 1}, \cdots, v_{s + K/2}\}$.
If $(i, j) \not \in B_{u_{s}}$, then $|i - j| > K/2$, which is contradiction to the definition of ring lattice $G$.
\fi
In our experiments, we show this by considering the starting graph, where each node belongs to exactly one clique of size 3.
This also requires us to modify the rewiring procedure, which we change accordingly.

We then build the MLN used in our experiments by adding the following weighted formulas to $\Psi$:
\begin{equation}
    \label{eq:ws_friends}
    \begin{aligned}
        &\infty, &\forall x\forall y: friends(x,y) \leftrightarrow Edge(x, y) \\
        &w, &smokes(x) \land friend(x,y) \rightarrow smokes(y),
    \end{aligned}
\end{equation}
where the wired edges in the Watts-Strogatz model forms a friendship network.

We also consider a simplified version of the Watts-Strogatz model, where we start with a ring lattice (or a different graph as above) and instead of rewiring edges, we introduce $M$ new edges, creating shortcuts in the graph.
Its MLN is
\begin{align*}
    &\infty, &\forall x\forall y: \lnot WiredEdge(x,y) \vee \lnot WiredEdge(y, x) \\
    &\infty, &EvidenceEdge(x, y) \rightarrow WiredEdge(x, y)
\end{align*}
with the cardinality constraint $|WiredEdge| = 2\cdot |\Delta| + M$.

The results of our experiments can be seen in Figure \ref{fig:ws}.

% In our paper, we compute the inference over the Watts-Strogatz graph encoded to a \Ctwo{} theory as well as over a simplified version of the model.
% In the simplified version of the model, we start with a ring lattice and instead of rewiring the edges, we only introduce $M$ shortcuts through the graph.
% We use binary evidence to enforce the starting graph to be a ring lattice or cliques of size 3 in our experiments.
% Since Forclift does not support counting quantifiers, we have to encode counting using binary evidence, which makes it underperform.
% The results of our experiments can be seen in Figure \ref{fig:ws}.

\begin{figure}[t]
\centering
\includegraphics[width=0.7\textwidth]{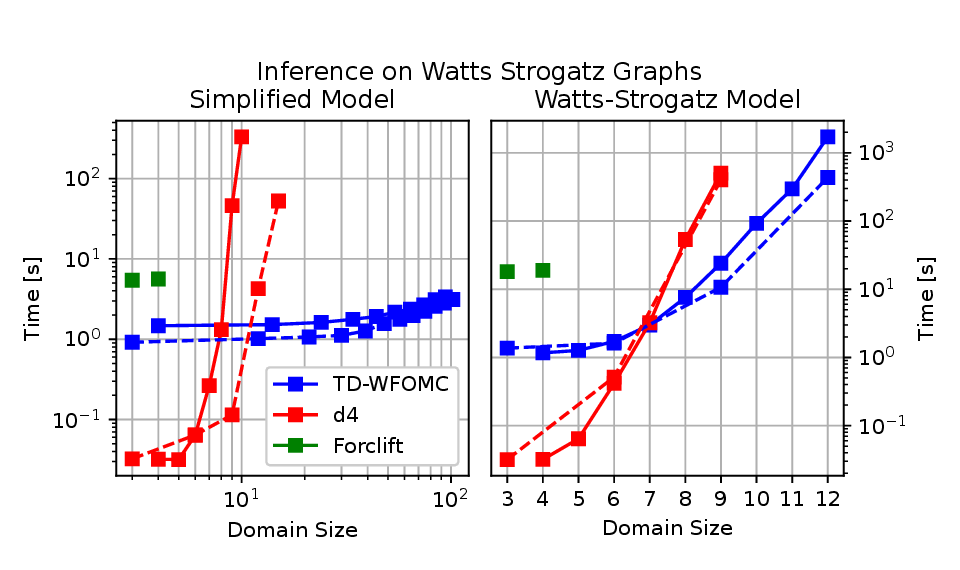} % Reduce the figure size so that it is slightly narrower than the column. Don't use precise values for figure width.This setup will avoid overfull boxes.
\caption{Inference times on Watts-Strogatz model, $K = 2$ and the simplified model, $K = 2$ and $M = \lfloor \frac{|\Delta| + 1}{2} \rfloor$. We use full line for a ring lattice starting graph and dashed line for cliques starting graph.}
\label{fig:ws}
%\end{figure}
\end{figure}

\section{New Results on the Stable Seating Arrangement Problem}

In the stable seating arrangement problem we have a group of agents $\mathcal{A} = [n]$ and a seating graph $G = (V_G, E_G)$.
In this and the following section, we use $[n]$ as a shorthand for $\{1, \cdots, n \}$.
Let $d$ be the maximum degree of $G$.
We use $N_G(v)$ to denote the neighborhood of $v \in V_G$.
Each agent needs to be assigned a different seat (a vertex of graph G) and we assume $n = |V_G|$.

Each agent $a \in \mathcal{A}$ has a \textit{preference} $p_a: \mathcal{A} \setminus \{a\} \mapsto \mathbb{R}$ of all other agents.
The preference $p_a(b)$ represents how much utility agent $a$ gets by sitting next to agent $b$.
A \textit{preference profile} is a collection of agent preferences, denoted by $\mathcal{P} = (p_a)_{a \in \mathcal{A}}$.
A \textit{class} of agents is a subset of agents $K \subseteq \mathcal{A}$ such that all agents in it share a common preference function $p_{K}: \mathcal{A} \mapsto \mathbb{R}$ and no agent distinguishes between others in the class.
A preference profile $\mathcal{P}$ is \emph{\textit{k}-class} if it can be partitioned into $k$ classes $K_1, K_2, \dots, K_k$ such that $\mathcal{A} = K_1 \cup K_2 \cup \dots \cup K_k$.
We denote the set of all classes by $\mathcal{K} = \{K_i | i \in [k]\}$.
We also use the notation $\mathcal{K}^* = \mathcal{K} \cup K_0$, where $K_0$ is a special class such that it holds that $p_{K_i}(K_0) = 0$ for all $i \in [k]$.
We use class $K_0$ to mark missing neighbors for seats in the seating graph that have fewer than $d$ neighbors.

An \emph{arrangement} is a bijection $\pi : \mathcal{A} \mapsto V_G$ which assigns each agent a seat in the seating graph $G$.
Given an arrangement $\pi$, the utility of each agent $a \in \mathcal{A}$ is $U_a(\pi) = \sum_{v \in N_G(\pi(a))}p_a(\pi^{-1}(v))$.
An agent $a$ \textit{envies} another agent $b$ if $U_a(\pi) < U_{a}(\pi')$ where $\pi'$ is obtained from $\pi$ by swapping the positions of agents $a$ and $b$.
An arrangement is called \emph{envy-free} if no agent envies another agent, and called \emph{stable} if no two agents envy each other.

\begin{definition}[Stable and Envy-Free Seating Arrangements]
A $k$-class stable dinner seating arrangement problem $SSA(\mathcal{A}, \mathcal{P}, G, \mathcal{K})$ is a problem of finding a stable seating arrangement $\pi$. $\#SSA(\mathcal{A}, \mathcal{P}, G, \mathcal{K})$ is the corresponding problem of counting the number of stable arrangements of $SSA(\mathcal{A}, \mathcal{P}, G, \mathcal{K})$.

$\#EFSA(\mathcal{A}, \mathcal{P}, G, \mathcal{K})$ is defined analogously with respect to envy-free seating arrangements.
\end{definition}

The complexity of finding stable seating arrangements in various settings has been studied \cite{Bullinger_Suksompong_2023,Bodlaender_Hanaka_2025,Roger_2023}.
\citet{Roger_2023} uses the notion of $k$-class stable seating arrangement, where each agent belongs to one of $k$ classes, and each class is a group of agents that share the same utility function and are also indistinguishable to each other in the group.
As established by \citet{Roger_2023}, finding a stable or envy-free seating arrangement for cyclic or path-shaped tables is solvable in polynomial time when the number of classes is fixed.

% In \Cref{thm:seatings} presented below, we extend this positive result to the realm of counting problems and more complex table shapes, demonstrating that counting the number of stable seating arrangements on bounded-treewidth bounded-degree seating graphs, which includes the cyclic or path-shaped tables, is also tractable, provided the number of classes remains fixed. Hence, \Cref{thm:seatings} also answers the open problem of the $2\times n$ grid table in \cite{Roger_2023}, as the seating graph has treewidth 2, and also generalizes the results to more complex table shapes.

In \Cref{thm:seatings} below, we extend this positive result to counting problems and more complex table configurations. Specifically, we show that counting the number of stable seating arrangements is tractable on seating graphs of bounded treewidth and bounded degree---conditions that include cycles and paths---provided the number of classes is fixed. This result also answers the open question posed in \cite{Roger_2023} concerning the $2 \times n$ grid table, whose seating graph has treewidth 2, and further generalizes the tractability to a broader class of table structures.

\begin{theorem}
    \label{thm:seatings}
    Counting the number of stable (or envy-free) seating arrangements is in time polynomial in the number of agents, if the problem is $k$-class and the seating graph has bounded treewidth and degree.
\end{theorem}

% It is known due to \citet{Roger_2023} that finding a stable or envy-free seating arrangement on tables that are cycles or paths is NP-complete.
% However, this problem becomes tractable if we fix the number of classes.

\begin{proof}
Let $\#SSA(\mathcal{A}, \mathcal{P}, G, \mathcal{K})$ and $\#EFSA(\mathcal{A}, \mathcal{P}, G, \mathcal{K})$ be the counting version of the stable and envy-free seating arrangements problems respectively.
We first show that there is an \FOtwo{} sentence $\Psi$, unary evidence $\mathcal{U}$, binary evidence $\mathcal{E}$, and a set of cardinality constraints $\mathcal{C}$, such that $\#SSA(\mathcal{A}, \\ \mathcal{P}, G, \mathcal{K})$ can be obtained from $\wfomc(\Psi \wedge \mathcal{U} \wedge \mathcal{E} \wedge \mathcal{C}, V_G, \weight, \negweight)$ in time polynomial in $|V_G|$.
%$WFOMC(\Psi \wedge \mathcal{U} \wedge \mathcal{E} \wedge \mathcal{C}, V_G, \weight, \negweight) \cdot \prod_{s \in \mathcal{K}} |s|!$, where $WFOMC(\Psi \wedge \mathcal{U} \wedge \mathcal{E} \wedge \mathcal{C}, V_G, \weight, \negweight)$ is domain liftable.

For a predicate $P$, we use $P/k$ to indicate that the arity of $P$ is $k$. The predicates in the formula $\Psi$ are:
\begin{itemize}
    \item $neighbor\_i/2$ ($i \in [d]$), which are used to encode the seating graph.
    The $neighbor\_i$ predicates are fully specified using closed-world evidence.
    For each vertex $u \in V_G$, we arbitrarily label all its outgoing edges so that no two edges share the same label $i$.
    This is useful when we later distinguish the individual neighbors of $u$.
    \item $class\_s/1$ ($s\in \mathcal{K}$), where $class\_s(x)$ indicates that the class of the agent sitting at vertex $x$ is $s$.
    \item $neighbor \_i\_is\_s/1$ ($i \in [d], \ s \in \mathcal{K}^*$), where $neighbor \_i\_is\_s(x)$ denotes that the $i$-th neighbor of the vertex $x$ is of class $s$. Recall that $\mathcal{K}^* = \mathcal{K} \cup K_0$ where $K_0$ is a special class such that $p_{s}(K_0) = 0$ for all $s \in \mathcal{K}$. We use $neighbor \_i\_is\_K_0(x)$ in case the degree $d_a$ of $x$ is less than the maximum degree $d$ of $G$.
    \item $envies/2$, where $envies(x,y)$ indicates that the agent at seat $x$ envies the agent at seat $y$.
\end{itemize}
We encode the seating graph to closed-world binary evidence $\mathcal{E}$.
Let
\begin{equation*}
  \mathcal{E} = \bigcup_{a \in V_G} \{neighbor\_i(a, b_i) \;|\; i \in [d_a], (a, b_i) \in E_G\},
\end{equation*}
where $d_a$ is the degree of $a$.
Recall that we label the outgoing edges of $a$ in such a way that $a$ is connected to each of its neighbors using a differently labeled edge.
\iffalse
$\mathcal{E} = \{neighbor\_i(a, b) \}$, where $(a, b) \in E_G$, such that it holds that for all $b < c$ where $b, c \in [n]$, if $neighbor\_i(a, b \in \mathcal{E})$ and $neighbor\_j(a, c) \in \mathcal{E}$, then $i < j$.
$\mathcal{E} = \{neighbor\_i(a, b) \}$, where $(a, b) \in E_G$, such that .
Also, for each vertex $a$ with degree $d_a < d$, it holds that if $neighbor\_i(a, b) \in \mathcal{E}$, then $i \in [d_a]$.
\fi
The unary evidence $\mathcal{U}$ is only necessary for seating graphs with at least one vertex $a$ with degree $d_a < d$. In such case, we define
\begin{equation*}
  \mathcal{U} = \bigcup_{a \in V_G} \{neighbor\_i\_is\_K_0(a) \;|\; d_a < i \le d\}.
\end{equation*}
%, where $d$ is the maximum degree of $G$ and $d_a$ is the degree of $a$.
Let $\Phi$ be a first-order sentence
%{\allowdisplaybreaks
\begin{align}
    \Phi =
    & \forall x : \bigvee_{s \in \mathcal{K}} class\_s(x) \label{ssa:phi:1}\\
    \land & \forall x: \bigwedge_{s \in \mathcal{K}, t \in \mathcal{K}, s \ne t} \lnot class\_s(x) \vee \lnot class\_t(x) \label{ssa:phi:2}\\
    \land & \forall x: \bigwedge_{i=1}^d \bigvee_{s \in \mathcal{K}^*} neighbor\_i\_is\_s(x) \label{ssa:phi:3}\\
    \land & \forall x: \bigwedge_{i=1}^d \ \bigwedge_{s,t \in \mathcal{K}^*, s \neq t} \Big(\lnot neighbor\_i\_is\_s(x) \vee \lnot neighbor\_i\_is\_t(x)\Big) \label{ssa:phi:4}\\
    \land & \bigwedge_{i \in [d], s \in \mathcal{K}} \forall x \forall y: (class\_s(y) \wedge neighbor\_i(x, y)) \to neighbor\_i\_is\_s(x). \label{ssa:phi:5}
\end{align}
%}

Let $\mathcal{C} = \left\{|class\_s| = |s| \big| s \in \mathcal{K}\right\}$ be a set of cardinality constraints and let $\weight(P) = \negweight(P) = 1$ for all predicates $P$ in $\Phi$.
Then $\wfomc(\Phi \wedge \mathcal{U} \wedge \mathcal{E} \wedge \mathcal{C},V_G, \weight, \negweight) \cdot \prod_{s \in \mathcal{K}}|s|!$ counts the number of possible seating arrangements (possibly unstable or not envy-free) for the set of agents $\mathcal{A}$ at a table given by the seating graph $G$.
%Both the binary evidence $\mathcal{E}$ and the unary evidence $\mathcal{U}$ are fixed and the predicates $neighbor\_i\_is\_s(x), i \in [d], s \in \mathcal{K}^*$ follow from the classes of agents at each seat (\ref{ssa:phi:3}, \ref{ssa:phi:4}, \ref{ssa:phi:5}). This means we count the permutations of classes of agents assigned to seats. We set the correct number of agents using $\mathcal{C}$ and add
Here, the term $\prod_{s \in \mathcal{K}}|s|!$ counts for each class the number of arrangements of the agents to the vertices of the class.
%This is precisely  the number of all seating arrangements (possibly unstable) from the definition in Section D.

Next, let
\begin{equation*}
  \begin{aligned}
  \mathcal{F} = \Big\{ & class\_s(x) \wedge \bigwedge_{i  \in [d], t_i \in \mathcal{K}^*} neighbor\_i\_is\_t_i(x)
     \land \bigwedge_{i \in [d], u_i \in \mathcal{K}^*} neighbor\_i\_is\_u_i(y) \\
    & \big| s\in \mathcal{K}, \sum_{i \in [d]} p_{s}(t_i) < \sum_{i \in [d]} p_s(u_i)\Big\}
  \end{aligned}
\end{equation*}
and
\begin{align*}
    \Gamma = \forall x \forall y:\;& \neg envies(x, y) \vee \neg envies(x, y) \\
    \land \forall x \forall y:\;& envies(x, y) \leftrightarrow \left(
 	\bigvee_{f \in \mathcal{F}} f(x, y)\right).
\end{align*}

Since $\Gamma$ filters out all unstable seating arrangements, we can see that if
$\Psi = \Phi \wedge \Gamma$, then
\begin{equation*}
\begin{aligned}
   \#SSA(\mathcal{A}, \mathcal{P}, G, \mathcal{K})
  = \wfomc(\Psi \wedge \mathcal{U} \wedge \mathcal{E} \wedge \mathcal{C}, V_G, \weight, \negweight) \cdot \prod_{s \in \mathcal{K}} |s|!.
\end{aligned}
\end{equation*}

Since the formula $\Psi$ is fixed and depends only on the number of classes and the maximum degree of the seating graph, we can say that due to results in Section 3.2, we prove that $\#SSA(\mathcal{A}, \mathcal{P}, G, \mathcal{K})$ can be computed in time polynomial in the number of agents if the maximum degree $d$ of the seating graph $G$ and the treewidth $t$ of $G$ are fixed.

We can use the same proof for computing the number of envy-free seating arrangements $\#EFSA(\mathcal{A}, \mathcal{P}, G, \mathcal{K})$ if instead of $\forall x \forall y: \neg envies(x ,y) \vee \neg envies(y, x)$ we require only $\forall x \forall y: \neg envies(x, y)$.
\end{proof}

\section{Conclusion}

We provide a novel approach to tackle the symmetry limitation in WFOMC by introducing an algorithm to compute WFOMC of \FOtwo{} and \Ctwo{} with cardinality constraints, unary evidence and binary evidence where the underlying Gaifman graph is of bounded treewidth in time polynomial in the domain size. The algorithm applies to the counting problem in combinatorics that can be encoded by WFOMC of such fragment, e.g., the stable seating arrangement problem. We hope this work inspires further advances to asymmetric lifted inference, and establishes a general counting method on bounded-treewidth graphs.
% analogous to Courcelle's theorem.

\section*{Acknowledgements}

V\'{a}clav Kůla is supported by the  Central Europe Leuven Strategic Alliance (CELSA) project {\em Towards Scalable Algorithms for Neuro-Symbolic AI}.
Yuanhong Wang is supported by National Natural Science Foundation of China (No.62506141). Ond\v{r}ej Ku\v{z}elka is supported by the Czech Science Foundation project 23-07299S ({\em Statistical Relational Learning in Dynamic Domains}).

\bibliographystyle{apalike}
\bibliography{ref}

\end{document}